\documentclass[12pt]{article}

\usepackage[pdftex]{graphicx}
\usepackage{latexsym}
\usepackage{epsfig}
\usepackage{changebar}
\usepackage{amssymb}
\usepackage{indentfirst}
\usepackage[english]{babel}
\usepackage[autostyle]{csquotes}
\usepackage{dsfont}
\usepackage{amsmath}
\usepackage{amsthm}

\newtheorem{theorem}{Theorem}

\begin{document}

\begin{titlepage}
	\vspace*{\fill}
	\begin{center}
		{\huge \textbf{Statistical Analysis of Binary Functional Graphs of the Discrete Logarithm}}\\[0.4cm]
		By\\[0.4cm]
		\large\textbf{Mitchell Orzech}\\[0.2cm]
		\large\textbf{Rose-Hulman Institute of Technology}\\[0.2cm]
		\large\textbf{May 19, 2016}
	\end{center}
	\vspace{1in}
	\vspace*{\fill}
\end{titlepage}

\pagestyle{plain}
\setcounter{page}{1}

\begin{abstract}
	The increased use of cryptography to protect our personal information makes us want to understand the security of cryptosystems. The security of many cryptosystems relies on solving the discrete logarithm, which is thought to be relatively difficult. Therefore, we focus on the statistical analysis of certain properties of the graph of the discrete logarithm. We discovered the expected value and variance of a certain property of the graph and compare the expected value to experimental data. Our finding did not coincide with our intuition of the data following a Gaussian distribution given a large sample size. Thus, we found the theoretical asymptotic distributions of certain properties of the graph.
\end{abstract}

\begin{section}{Introduction}

	With the increase in computational power continuing to rise, it has become important to understand the security of cryptosystems used everyday in order to ensure that the privacy of individuals is protected during bank transactions or private communications. Many of these cryptosystems rely on the common mapping
	\begin{equation} \label{discExp}
		x \mapsto g^x \text{ mod } n
	\end{equation}
	where the function takes the set $\{1, \ldots, n-1\} \mapsto \{1, \ldots, n-1\}$ and the gcd$(g,n) = 1$. Functions that resemble this type of mapping are known as discrete exponentiation mappings and they are known to exist in cryptographic schemes such as ElGamal encryption, Diffie-Hellman key exchange, and the Digital Signature Algorithm. Additionally, it is known that if $n$ is a prime and $g$ is a primitive root modulo $n$, then a map of the form in \eqref{discExp} has an inverse known as the \textit{discrete logarithm}. Computing the discrete logarithm appears relatively difficult, which is why many cryptosystems use \eqref{discExp} in their encryption schemes.

	When analyzing the security of cryptographic schemes, we want to ensure that the adversary cannot obtain the secret information used in the schemes in a method better than at random. Therefore, when looking at the graphs underlying the schemes, it makes sense to compare them to a random graph of the same form. Many statistics for random graphs are described in \cite{FO}. In this paper, we will focus on the statistics of binary functional graphs of \eqref{discExp}. We will focus on statistics pertaining to the rho length of these graphs, as tail and cycle length have been discussed in \cite{NL}. We conjecture and then prove that the average rho length is the sum of the average tail and cycle lengths found in \cite{NL}. We find the variance of the average rho length of a random binary functional graph. Then we gather experimental data on multiple binary functional graphs of \eqref{discExp} for various values of $n$. The results we obtain from the data lead us to look at the distributions of the different sizes of cycle, tail, and rho lengths, for which we get asymptotic results.

\end{section}

\begin{section}{Terminology and Background}
	Throughout this paper, we will let $p$ stand for an odd prime. The mappings we focus on are
	$$f: S = \{1, \ldots, p - 1\} \mapsto S$$
	of the form \eqref{discExp} where $n = p$ and gcd$(g,p) = 1$. Our focus will be binary functional graphs of \eqref{discExp}. A functional graph is a directed graph where each node has exactly one edge emanating from it. An $m$-ary functional graph is a functional graph where each node has exactly 0 or $m$ edges coming into it. A binary functional graph is an $m$-ary functional graph with $m = 2$. One can easily see that by the definition of binary functional graphs, the number of nodes with 0 edges coming into it is the same as the number of nodes with 2 edges coming into it. As described in \cite{DCJH}, we have a theorem that explains the relationship of $g$ to binary functional graphs.
	\begin{theorem}
		Let $p$ be a fixed prime and let $m$ be any positive integer that divides $p - 1$. Then as $g$ ranges from $1$ to $p - 1$, there are $\phi((p-1)/m)$ different functional graphs which are $m$-ary produced by graphs of the form \eqref{discExp}. Futhermore, if $r$ is any primitive root modulo $p$ and $g \equiv r^a$ mod $p$, then the values of $g$ that produce an $m$-ary graph are those where gcd$(a,p-1) = m$.
	\end{theorem}
	In the above theorem $\phi$ is Euler's phi (or totient) function. By Theorem 1, we see that the values of $g$ that give us binary functional graphs are those that are squares of primitive roots modulo $p$. An example of a binary functional graph produced by \eqref{discExp} is given in Figure \ref{fig:graph}. Now that we know what characterizes a binary functional graph, we can look at tail, cycle, and rho lengths.

	\begin{figure}[t]
		\begin{center}
			\includegraphics[scale=0.6]{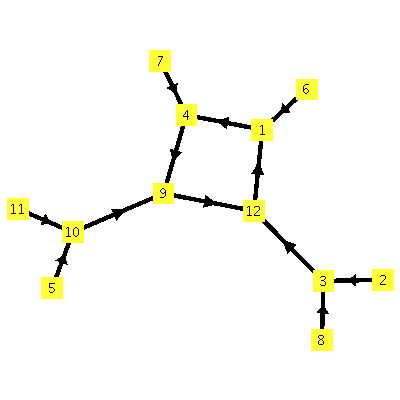}
		\end{center}
		\caption{A binary functional graph produced by \eqref{discExp} where $n = 13$ and $g = 4$.}
		\label{fig:graph}
	\end{figure}

	Let the ordered pair $(x, f(x))$, $x,f(x) \in S$, represent the traversal along an edge of the binary functional graph $f: S \mapsto S$ from node $x$ to node $f(x)$. Let $x_0$ be a random starting node in $f$. Since $f$ is a directed mapping and the cardinality of $S$ is $p - 1$, it follows from the pigeonhole principle that there must exist a point where $f^i(x_0) = f^j(x_0)$, $i \neq j$, after at most $p$ iterations. Suppose $i < j$. The tail length (as seen from node $x_0$), is the number of edges from $x_0$ to $f^i(x_0)$, which is $i$. The cycle length (as seen from node $x_0$) is the number of edges from $f^i(x_0)$ to itself, which is $j - i$. The rho length (as seen from a node $x_0$) is the sum of the tail and cycle length, which is $j$. Looking at Figure \ref{fig:graph}, we observe that node 5 sees a cycle length of 4, a tail length of 2, and a rho length of 6. Primarily we are concerned with the average and variance of these properties. If we take the sum of the cycle lengths, tail lengths, and rho lengths across every node, we get that the total cycle length is $12(4) = 48$, the total tail length is $4(0) + 4(1) + 4(2) = 12$, and the total rho length is $48 + 12 = 60$. Therefore, if we divide by the number of nodes, we get that the average cycle length is 4, the average tail length is 1, and the average rho length is 5. Also, the cycle length has variance 0, the tail length has variance 1, and the rho length has variance 1.

	We now state the major theorem of this paper:
	\begin{theorem} \label{expectedaveRL}
		The expected value for the rho length as seen from a random node in a random binary functional graph of size $n$ is the sum of the expected values for the tail and cycle length as seen from a random node in a random binary functional graph of size $n$.
	\end{theorem}
	This idea comes from the fact that the rho length is the sum of the tail and cycle length, and since expected value is a linear operation, Theorem \ref{expectedaveRL} should hold. The reason that the expected value of the rho length is of importance is because it tells us when we expect to see our first repeated value from a random starting node. In terms of cryptography, this information allows us to predict on average how many ``messages" we need to see before we see a reuse of a key. We prove this theorem, along with other resulting facts, by techniques used in \cite{DCJH}, \cite{FO}, and \cite{NL} using exponential generating functions on graphs. An exponential generating function, $\eta(z)$, is a function that describes the sequence of real numbers $(\eta_0, \eta_1, \eta_2, \ldots)$ with the power series
	$$\eta(z)=\sum_{i=0}^\infty \frac{\eta_i}{i!} z^i.$$
	The generating functions of interest for binary functional graphs as stated in \cite{DCJH} and \cite{NL} are
	\begin{align*}
		f(z) &= e^{c(z)} = \frac{1}{1-zb(z)} \\
		c(z) &= \ln \frac{1}{1-zb(z)} \\
		b(z) &= z + \frac{1}{2} zb^2(z)
	\end{align*}
	where $f$ generates the number of binary functional graphs, $c$ generates the number of components, and $b$ generates the number of binary trees. Solving for $b(z)$ in the implicit equation, we get
	\begin{align}
		f(z) &= \frac{1}{\sqrt{1-2z^2}} \\
		c(z) &= \ln \frac{1}{\sqrt{1-2z^2}} \\
		b(z) &= \frac{1 - \sqrt{1-2z^2}}{z}.
	\end{align}
	It will be these values that we use to construct our generating functions for the properties of binary functional graphs we want to analyze. Another fact that will be needed is the number of binary functional graphs of size $n$. This will be used to normalize results to average over all binary functional graphs. For convenience later, we will use the result $g(n)$ found in \cite{NL} which is the number of binary functional graphs of size $n$ divided by $n!$:
	$$ g(n) = \frac{2^\frac{n}{2} \Gamma \left( \frac{n}{2} + \frac{1}{2} \right)}{ \sqrt{\pi} \Gamma \left( \frac{n}{2} + 1 \right) }. $$
	where $\Gamma(n) = \int_{x=0}^\infty e^{-x} x^{n-1} \, dx$. We will also need a function asymptotically equal to $g(n)$, which in \cite{DCJH} is found to be:
	$$ g^*(n) = \frac{2^\frac{n}{2} \sqrt{2}}{\sqrt{\pi n}}.$$ 
	Lastly, since recurrence relations appear in the proofs of some results, we define them here: A recurrence relation is an equation that defines a sequence of values where the next term in the sequence is defined as a function of some of the preceding terms.

\end{section}

\begin{section}{Theoretical Results of Rho Length}
	As stated above, the average rho length is a desired property for these graphs. The result is restated in its expression form in Theorem \ref{aveRL}.
	\begin{theorem} \label{aveRL}
		The expected value for the rho length as seen from a random node in a random binary functional graph of size $n$ is
		$$ \frac{(n+1)\left( n \sqrt{\pi} \Gamma\left( \frac{n}{2} \right) - 2\Gamma\left( \frac{n}{2} + \frac{1}{2} \right) \right)}{2 n \Gamma\left( \frac{n}{2} + \frac{1}{2} \right)}. $$
	\end{theorem}

	\begin{proof}
		As in \cite{DCJH} and \cite{FO}, the generating function needs to be defined with a parameter $u$ to mark the nodes of interest. Let $\Xi(z)$ be the generating function for the total rho length over all binary functional graphs. Then $\Xi(z)$ can be defined as
		\begin{equation}\label{aveRLGF}
			\Xi(z) = \frac{\partial}{\partial u} \left[ \frac{1}{1-zb(z)} \frac{1}{1-uzb(z)} uz \left( \frac{ub(z)}{1-uzb(z)} + b(z) \right) \right]_{u = 1}.
		\end{equation}
		In \eqref{aveRLGF}, $\frac{1}{1-zb(z)}$ is for the unmarked components, $\frac{1}{1-uzb(z)}$ is for the unmarked trees in the marked component, $uz$ is for marking the node where the tail meets the cycle, $\frac{ub(z)}{1-uzb(z)}$ is for making edges along one tail and is derived from the equation
		$$ \beta(z,u) = z + \frac{1}{2} z b^2(z) + uzb(z) \beta(z,u)$$
		from \cite{DCJH}, and $b(z)$ is for the nodes with no tail lengths. Then using a package called \texttt{gfun} in Maple, we find a differential equation satisfied by $\Xi(z)$ by using the function \texttt{holexprtodiffeq($\Xi(z)$, $y(z)$)}. When we apply this function, we get the differential equation
		\begin{equation} \label{diffeqAveRL}
			-4z^3 + 6z + (24z^5-24z^3+6z)y(z) + (8z^6-12z^4+6z^2-1)y'(z)=0, \quad y(0) = 0.
		\end{equation}
		We then acquire a recurrence relation whose solutions are the coefficients to the solution to \eqref{diffeqAveRL} by using the function \texttt{diffeqtorec(\eqref{diffeqAveRL}, $y(z)$, $u(n)$)} in \texttt{gfun}, resulting in
		\begin{gather*}
			0 = (24 + 8n)u(n) + (-48-12n)u(n+2) + (6n+30)u(n+4)+(-n-6)u(n+6), \\
			u(0) = u(1) = u(3) = u(5) = 0, \\
			u(2) = 3, \\
			u(4) = \frac{25}{2}.
		\end{gather*}
		This recurrence relation for $u(n)$ can be simplified by letting $n = 2k$, resulting in
		\begin{gather*}
			0 = (24 + 16k)u(k) + (-48-24k)u(k+1)+(12k+30)u(k+2)+(-2k-6)u(k+3), \\
			u(0) = 0, \\
			u(1) = 3, \\
			u(2) = \frac{25}{2}.
		\end{gather*}
		Solving this recurrence relation for $u(k)$ in Maple using \texttt{rsolve()} results in
		\begin{equation} \label{ksolAveRL}
			u(k) = 2^k(1 + 2k) - \frac{2^{k+1} \Gamma\left(\frac{3}{2} + k\right)}{\sqrt{\pi} \Gamma(k+1)}.
		\end{equation}
		Equation \eqref{ksolAveRL} then needs to be normalized. First it must be multiplied by $n!$ to get the correct parameter value $c$ since $\frac{c}{n!}$ is the coefficient of $z^n$ in $\Xi(z)$.
		Equation \eqref{ksolAveRL} also needs to be divided by the total number of binary functional graphs to get the expected value, and by $n$ to get the result from a random node in the graph. Since $g(n)$ is already divided by $n!$, the result for Theorem \ref{aveRL} can be obtained by taking the solution to the recurrence relation and dividing by $ng(n)$. Since we transformed the recurrence relation from $u(n)$ to $u(k)$, we first need to transform back to $u(n)$ by letting $k = \frac{n}{2}$ in $u(k)$. Then the result of Theorem \ref{aveRL} is given by 
		$$\frac{u(n)}{ng(n)}.$$
	\end{proof}

	To compare Theorem \ref{aveRL} with Theorem \ref{expectedaveRL}, we first need to state the results for average cycle and tail lengths found in \cite{NL}:

	\begin{theorem} \label{cltl}
		The expected values for the cycle length and tail length as seen from a random node in a random binary functional graph of size $n$ are
		\begin{align*}
			\text{Cycle Length:}& \qquad \frac{\sqrt{\pi} \Gamma(\frac{n}{2} + 1)}{2 \Gamma(\frac{n}{2} + \frac{1}{2})} \\
			\text{Tail Length:}& \qquad \frac{\sqrt{\pi} \Gamma(\frac{n}{2} + 2) - n \Gamma(\frac{n}{2} + \frac{1}{2}) - \Gamma(\frac{n}{2} + \frac{1}{2})}{n \Gamma(\frac{n}{2} + \frac{1}{2})}.
		\end{align*}
	\end{theorem}

	The reader can confirm Theorem \ref{expectedaveRL} by taking the sum of the results from Theorem \ref{cltl} and comparing to the result of Theorem \ref{aveRL}. When finding the expected value of a result, we usually also want the associated variance. The result for the variance of the average rho length is listed in Theorem \ref{varaveRL}.
	\begin{theorem}\label{varaveRL}
		The variance of the average rho length of a random binary functional graph of size $n$ is
		$$-\frac{ n^4 \pi \Gamma \left( \frac{n}{2} \right)^2 + 2 n^3 \pi \Gamma \left( \frac{n}{2} \right)^2 + 2 n^3 \sqrt{\pi} \Gamma \left( \frac{n}{2} \right) \Gamma \left( \frac{n}{2} + \frac{1}{2} \right) + n^2 \pi \Gamma \left( \frac{n}{2} \right)^2 + 2n^2 \sqrt{\pi} \Gamma \left( \frac{n}{2} \right) \Gamma \left( \frac{n}{2} + \frac{1}{2} \right)}{4 n^2 \Gamma \left(\frac{n}{2} + \frac{1}{2} \right)^2}$$
		$$-\frac{-2n^3 \Gamma \left( \frac{n}{2} + \frac{1}{2} \right)^2 - n \sqrt{\pi} \Gamma \left( \frac{n}{2} \right) \Gamma \left( \frac{n}{2} + \frac{1}{2} \right) - 6 n^2 \Gamma \left( \frac{n}{2} + \frac{1}{2} \right)^2 - 3 n \Gamma \left( \frac{n}{2} + \frac{1}{2} \right)^2 + \Gamma \left( \frac{n}{2} + \frac{1}{2} \right)^2}{n^2 \Gamma \left( \frac{n}{2} + \frac{1}{2} \right)^2}$$
	\end{theorem}
	\begin{proof}
		Using the fact that the variance of a set of data is $\frac{1}{N} \left( \sum_{i=1}^{N} x_i^2 \right) - \bar{x}^2$, where $x_i$ are the individual data points and $\bar{x}$ is the mean of the data, all we need is a generating function for $\sum_{i=1}^{N} x_i^2$ and the result from Theorem \ref{aveRL} to get the desired result of Theorem \ref{varaveRL}. The generating function for $\sum_{i=1}^N x_i^2$ can be found by using the technique described in \cite{NL}. In summary, the method is mark the nodes of interest with $u$, differentiate with respect to $u$, multiply by $u$ to correct for the power of $u$, differentiate with respect to $u$ again, and then plug in $u = 1$. This technique gives the desired value of the parameter squared.

		Let $\Xi^*(z)$ be the generating function for the total rho length squared. Using the same technique as in \cite{NL} and the generating function from the proof of Theorem \ref{aveRL}, we get
		$$\Xi^*(z) = \frac{\partial}{\partial u} \left\{ u \left[ \frac{\partial}{\partial u} \left( \frac{1}{1-zb(z)} \frac{1}{1-uzb(z)} uz \left( \frac{ub(z)}{1-uzb(z)} + b(z) \right) \right) \right] \right\}_{u = 1}. $$
		With the \texttt{gfun} package in Maple, we are able to find a differential equation satisfied by $\Xi^*(z)$ using the function \texttt{holexptodiffeq($\Xi^*(z)$, $y(z)$)}, resulting in
		\begin{gather} 
			8z^5 + 40z^3 + 50z + (48z^7 + 168z^5 - 204z^3+54z)y(z) \nonumber \\
			+ (16z^8 + 16z^6 - 48z^4 + 28z^2-5)y'(z), \label{diffeqVarAveRL} \\
			y(0) = 0. \nonumber
		\end{gather}
		We then acquire a recurrence relation whose solutions are the coefficients to the solution to \eqref{diffeqVarAveRL} by using the function \texttt{diffeqtorec(\eqref{diffeqVarAveRL}, $y(z)$, $u(n)$)} in \texttt{gfun}, resulting in
		\begin{gather*}
			0 = (48+16n)u(n)+(200+16n)u(n+2)+(-396-48n)u(n+4) \\
			+(222+28n)u(n+6)+(-5n-40)u(n+8), \\
			u(0) = u(1) = u(3) = u(5) = u(7) = 0, \\
			u(2) = 5, \\
			u(4) = \frac{59}{2}, \\
			u(6) = \frac{227}{2}.
		\end{gather*}
		This recurrence for $u(n)$ can be simplified by letting $n = 2k$, resulting in
		\begin{gather*}
			(48+32k)u(k)+(200+32k)u(k+1)+(-396-96k)u(k+2) \\
			+(222+56k)u(k+3)+(-10k-40)u(k+4), \\
			u(0) = 0, \\
			u(1) = 5, \\
			u(2) = \frac{59}{2}, \\
			u(3) = \frac{227}{2}.
		\end{gather*}
		Solving this recurrence relation for $u(k)$ in Maple using \texttt{rsolve()} results in
		\begin{equation} \label{ksolVarAveRL}
			u(k) = -2^k(5 + 6k)+\frac{2^{k+3} \Gamma\left(\frac{3}{2}+k\right) \left(k+\frac{5}{4}\right)}{\sqrt{\pi} \Gamma(k+1)}.
		\end{equation}
		Equation \eqref{ksolVarAveRL} then needs to be normalized. We need to divide by $n$ to get the average of the squared sums of the individual data required for the variance calculation, and by $g(n)$ to normalize by the number of graphs. We can ignore multiplying by $n!$ because $g(n)$ is already divided by $n!$, as described in the proof of Theorem 3. Therefore, dividing the solution to the recurrence relation by $ng(n)$ and subtracting the result for the mean from Theorem \ref{aveRL} squared gives the desired result. Before we can do that, we need to get the correct solution to the recurrence relation, $u(n)$. We get $u(n)$ by letting $k = \frac{n}{2}$ in $u(k)$. Then if we let $\mu$ be the mean value in Theorem \ref{aveRL}, the result of Theorem \ref{varaveRL} comes from
		$$\frac{u(n)}{ng(n)} - \mu^2.$$
	\end{proof}

	Once we have the expected value and variance for all binary functional graphs, we want to compare this to our particular graph \eqref{discExp} and see if this graph behaves like random binary functional graphs.

\end{section}

\begin{section}{Experimental Results}
	To compare the average rho length of binary functional graphs of \eqref{discExp} with the value obtained in Theorem \ref{aveRL}, we generate all $\phi \left( \frac{p-1}{2} \right)$ graphs. An individual data point in our sample is the result acquired from a graph of one prime across all $g$. Our null hypothesis for testing is that the sample behaves like a random sample from a normal distribution, and our alternative hypothesis is that the sample does not behave like a random sample from a normal distribution. In this paper, we will normalize each data point we obtain in the following fashion:
	$$\frac{(\bar{x} - \mu)\sqrt{n}}{s},$$
	where $\bar{x}$ is the experimental mean from one prime, $\mu$ is the theoretical mean from Theorem \ref{aveRL}, $n = \phi \left( \frac{p-1}{2} \right)$ is the number of trials, and $s$ is the experimental standard deviation from one prime. This is done to convert the data to a standard normal. We test the hypotheses by using two normality tests (Shapiro-Wilk and Anderson-Darling) and a normality plot to see if there is a resulting normal distribution after normalizing the data. Obtaining sample data for 600 consecutive primes starting with $100003$ and using modified code from \cite{DCJH} and \cite{NL} results in the following data for the average rho length: \\
	\begin{minipage}[t][4cm][c]{0.48\textwidth}
		\begin{center}
			Shapiro-Wilk normality test:
			$$\text{$p$-value} = 0.03197$$
			\begin{center}
				(reject at $\alpha = 0.05$)
			\end{center}
		\end{center}
	\end{minipage}
	\begin{minipage}[t][4cm][c]{0.48\textwidth}
		\begin{center}
			Anderson-Darling normality test:
			$$\text{$p$-value} = 0.00139$$
			\begin{center}
				(reject at $\alpha = 0.05$)
			\end{center}
		\end{center}
	\end{minipage}
	\begin{center}
		\includegraphics[scale=.55]{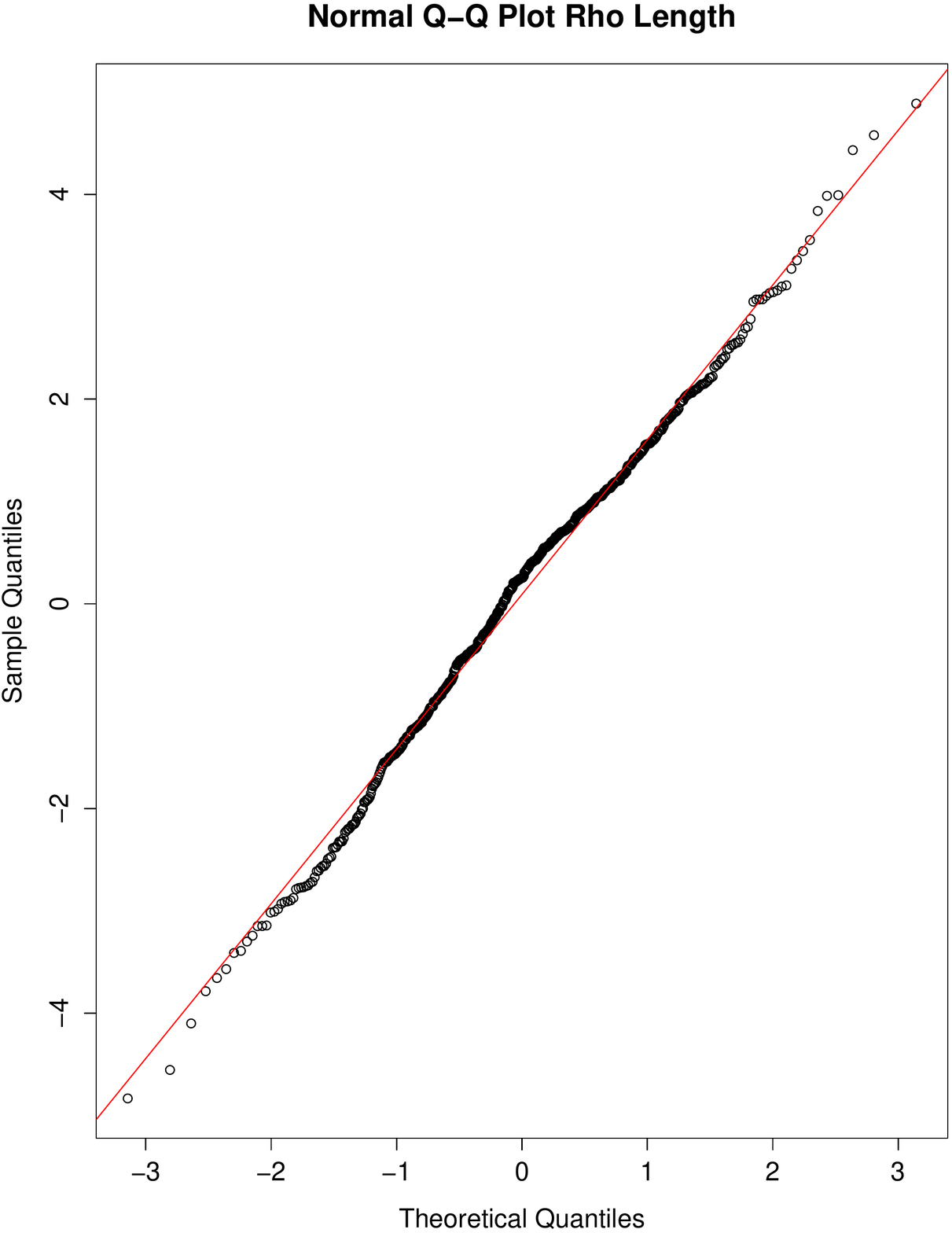}
	\end{center}

	As can be seen by the results, the normalized data does not follow a normal distribution. Therefore, we decided to look at the average cycle and tail lengths for these 600 primes; previous results were acquired in \cite{NL} on a much smaller set of primes. 
	For the average cycle length, we get \\
	\begin{minipage}[t][4cm][c]{0.48\textwidth}
		\begin{center}
			Shapiro-Wilk normality test:
			$$\text{$p$-value} = 0.5655$$
			\begin{center}
				(fail to reject at $\alpha = 0.05$)
			\end{center}
		\end{center}
	\end{minipage}
	\begin{minipage}[t][4cm][c]{0.48\textwidth}
		\begin{center}
			Anderson-Darling normality test:
			$$\text{$p$-value} = 0.6807$$
			\begin{center}
				(fail to reject at $\alpha = 0.05$)
			\end{center}
		\end{center}
	\end{minipage}
	\begin{center}
		\includegraphics[scale=.55]{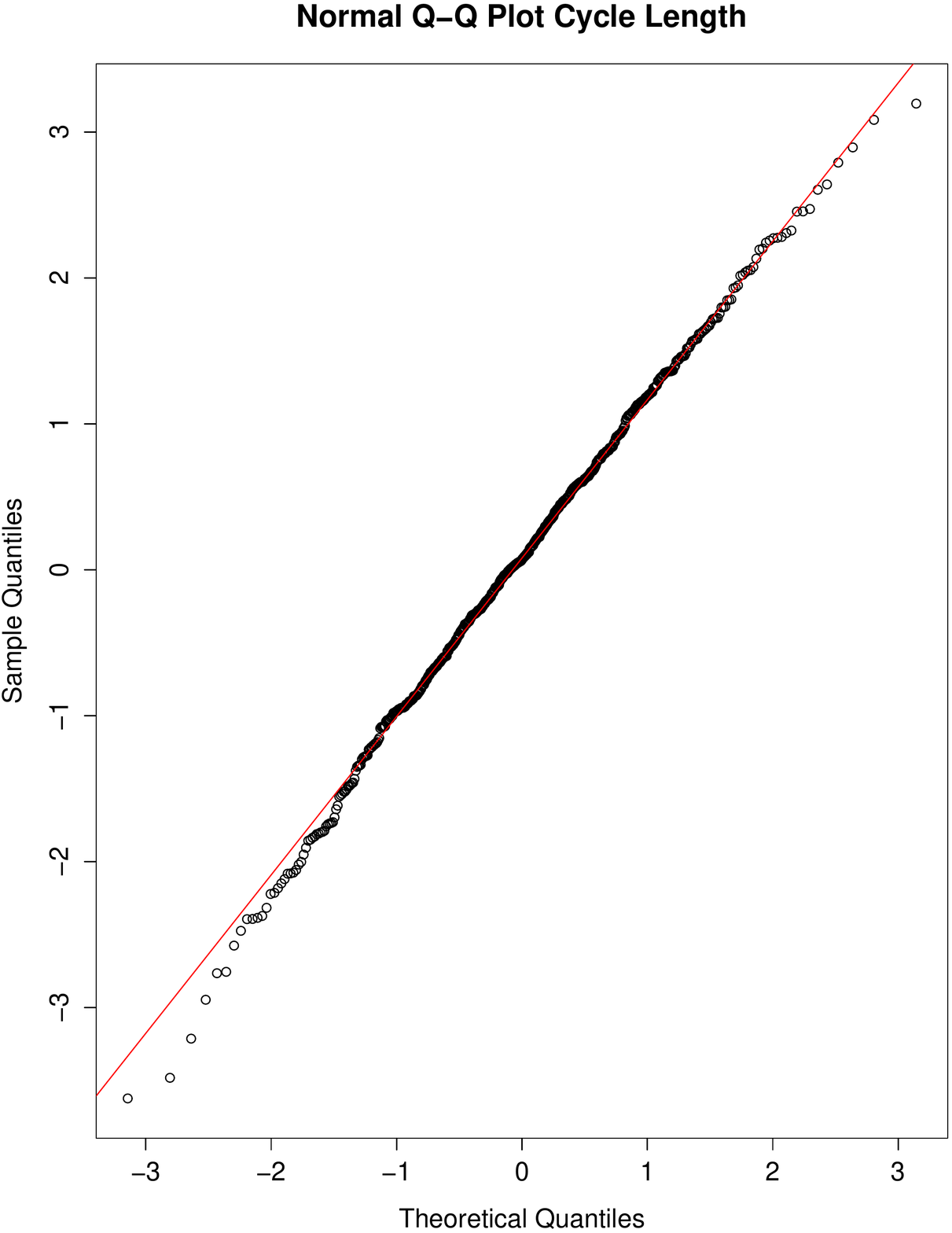}
	\end{center}
	\newpage
	For the average tail length, we get \\
	\begin{minipage}[t][4cm][c]{0.48\textwidth}
		\begin{center}
			Shapiro-Wilk normality test:
			$$\text{$p$-value} = 0.4075$$
			\begin{center}
				(fail to reject at $\alpha = 0.05$)
			\end{center}
		\end{center}
	\end{minipage}
	\begin{minipage}[t][4cm][c]{0.48\textwidth}
		\begin{center}
			Anderson-Darling normality test:
			$$\text{$p$-value} = 0.5529$$
			\begin{center}
				(fail to reject at $\alpha = 0.05$)
			\end{center}
		\end{center}
	\end{minipage}
	\begin{center}
		\includegraphics[scale=.55]{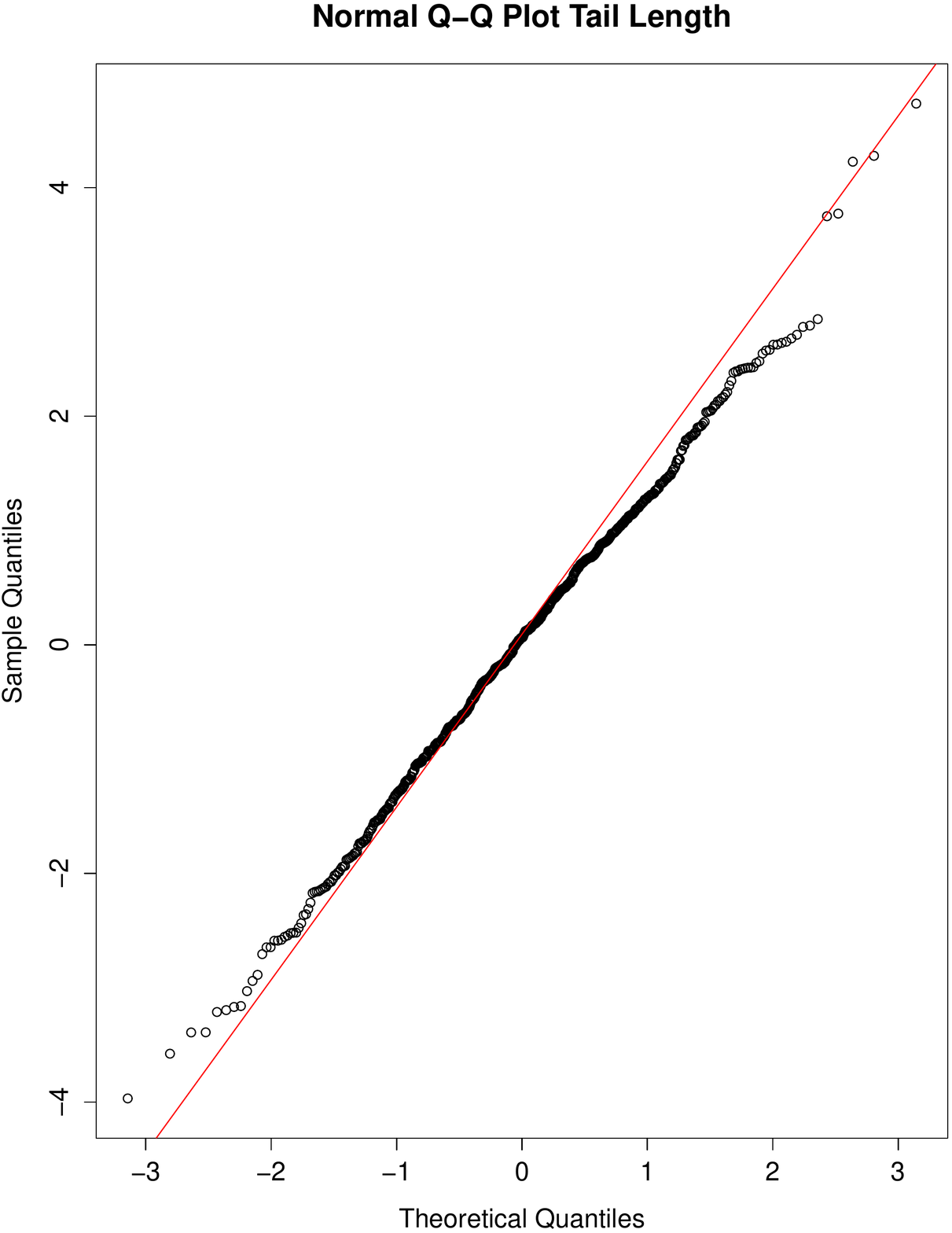}
	\end{center}

	Therefore, we see that the average cycle and tail lengths could be normally distributed after being normalized whereas the average rho length is not. Once we discovered this, we wanted to know more about their individual distributions. Thus, we stopped our experimental testing and moved towards finding asymptotic distributions for cycle, tail, and rho lengths.

\end{section}

\begin{section}{Theoretical Distributions}
	Implementing the same procedure as in \cite{FO}, we get the following asymptotic distributions:
	\begin{theorem} \label{distProps}
		For any fixed $r$, the parameters number of cycle lengths of size $r$, number of tail lengths of size $r$, and number of rho lengths of size $r$ have the asymptotic mean values:
		\begin{align*}
			\text{$r$-Cycle Lengths:}& \quad \sqrt{\frac{\pi n}{2}} - r \\
			\text{$r$-Tail Lengths:}& \quad \sqrt{\frac{\pi n}{2}} - r + 1 \\
			\text{$r$-Rho Lengths:}& \quad \quad r
		\end{align*}
	\end{theorem}
	\begin{proof}
		Since the method is the same for proving each result, we will focus our attention on proving $r$-rho lengths since we already have presented the generating function for generating the total rho length over all binary functional graphs. The generating functions for cycle and tail lengths of binary functional graphs can be found in \cite{DCJH}.

		Let $\delta(z)$ be the exponential generating function for the $r$-rho lengths over all binary functional graphs. Then $\delta(z)$ can be written as
		\begin{align*}
			\delta(z) &= \frac{\partial}{\partial u} \left\{ \frac{1}{1-zb(z)}\left[ \frac{1}{1-zb(z)} \left( \frac{zb(z)}{1-zb(z)} + zb(z) \right. \right. \right. \\
					  &\quad \left. \left. \left. {}+ (u-1)(r-1)(zb(z))^{r-1} + (u-1)(zb(z))^r \vphantom{\frac12} \right) \vphantom{\frac12} \right] \vphantom{\frac12} \right\}_{u=1}.
		\end{align*}
		Above, $(u-1)(r-1)(zb(z))^{r-1}$ comes from the expansion of the product $\frac{1}{1-zb(z)}\left(\frac{zb(z)}{1-zb(z)}\right)$, and $(u-1)(zb(z))^r$ comes from the expansion of the product $\frac{1}{1-zb(z)}(zb(z))$. The reason for the $r-1$ expressions in $(u-1)(r-1)(zb(z))^{r-1}$ is due to the fact that we mark the node where the tail meets the cycle.
		We then use the function \texttt{equivalent($\delta(z)$, $z$, $n$, 1)} in Maple from the package \texttt{algolib}\footnote{http://algo.inria.fr/libraries/} to convert the generating function into a first order asymptotic form of coefficients in terms of $n$, resulting in
		$$\frac{r 2^\frac{n}{2} \sqrt{2}}{\sqrt{ \pi n}}.$$
		The reason we used this approach is because we could not find a closed form differential equation that was satisfied by $\delta(z)$. Next we need to normalize this result by multiplying by $n!$ to get the correct parameter value and then dividing by the number of binary functional graphs. We want to divide by the asymptotic form of the number of binary functional graphs since we have an asymptotic approximation. Therefore, the desired normalized result in Theorem \ref{distProps} can be found by dividing by $g^*(n)$ since it is the asymptotic total number of binary functional graphs divided by $n!$. (Note: For cycle and tail length, we found a second order asymptotic form of coefficients since it was where the first occurrence of an $r$ term appeared.)
	\end{proof}

	This result matched what we were seeing; we expect to see many small cycle and tail lengths, but when it comes to a small rho length there are only a few ways it can be achieved. For example, the possible ways to see a rho length of 3 are a cycle length of 3, a cycle length of 2 and a tail length of 1, and a cycle length of 1 and a tail length of 2.

\end{section}

\begin{section}{Discussion and Future Work}
	From the linearity of expected value, we knew the average rho length of binary functional graphs using the results for average cycle and tail lengths in \cite{NL}. In this paper, we were able to prove the result using an exponential generating function, confirming that we had the correct generating function when it came to finding the variance in the average rho length. In \cite{NL}, it was found that the normalized average cycle and tail lengths appeared to follow normal distributions. Therefore, we investigated the normality for the normalized average rho length and found that we reject normality for both Shapiro-Wilk and Anderson-Darling normality tests on the normalized data. This led us to re-evaluate whether normalized average cycle and tail lengths were distributed normally on our larger data set. There are other normality tests that could be used, especially ones that possibly put more of a bias on outliers. The conflict in results led us to finding the asymptotic distributions to the average cycle, tail, and rho lengths, which matched our conjectures from small examples.

	There is much future work to be done in analyzing binary functional graphs. One thing would be to look at the normalized experimental variances for the average rho lengths and compare to the results found in \cite{NL} for the average cycle and tail lengths, which showed that the experimental results did not coincide with the theoretical results. We believe this is due to the fact that for a given cycle length $c$ in a binary functional graph, there are at least $2c$ nodes that see that cycle length. Similarly, there are at least $c$ nodes that see a tail length of one. A possible way to get around these discrepancies would be to sample one random node from all possible graphs and then measure the variance. 
	
	It would also be beneficial to test the experimental distributions against the asymptotic distributions. In this paper, we took second order asymptotic approximations for the cycle and tail lengths and first order for the rho length; it is not known how well these asymptotics correlate with the actual data. 
	
	Another interesting property to look at would be the component size. Since only one cycle exists in any component of a particular graph under a certain $g$, we know the number of nodes to see the cycle length of that component is at least the component size of that component; there may be more components that have the same cycle length. Therefore, there is definitely some correlation between component size and cycle length as seen from a node.

\end{section}

\begin{section}{Acknowledgements}
	I would like to thank Dr. Josh Holden, my advisor, for his constant guidance and excellent advice throughout this thesis. I would also like to thank Dr. Eric Reyes for all his guidance and help when it came to using R and interpreting the statistical results.
\end{section}

\newpage
\thispagestyle{empty}


\begin{thebibliography}{99}

\bibitem{DCJH} D. Cloutier and J. Holden, ``Mapping the discrete logarithm". \textit{Involve}, 3:197-213, 2010.

\bibitem{FO} P. Flajolet and A. Odlyzko, ``Random mapping statistics". \textit{Advances in cryptology}, pp. 329-354, Houthalen, Belgium, 1989. Edited by A. J. Menezes and S.A. Vanstone, Lecture Notes in Comput. Sci. \textbf{434}, Springer, Berlin, 1990. MR 1083961 Zbl 0747.05006

\bibitem{NL} N. Lindle, ``A Statistical Look at Maps of the Discrete Logarithm". \textit{Mathematical Sciences Technical Reports (MSTR)}, Paper 35, 2008.













\end{thebibliography}
\end{document}